\documentclass[a4paper,UKenglish,numberwithinsect]{lipics-v2016}
 
\usepackage{microtype}

\usepackage{nestra}

\DOIPrefix{}

\title{Alignment Elimination from Adams' Grammars}

\author[1]{Härmel Nestra}
\affil[1]{Institute of Computer Science, University of Tartu, J. Liivi 2,
50409 Tartu, Estonia\\
  \texttt{harmel.nestra@ut.ee}}
\authorrunning{H. Nestra}

\Copyright{Härmel Nestra}



\begin{document}

\maketitle

\begin{abstract}
Adams' extension of parsing expression grammars enables specifying 
indentation sensitivity using two non-standard grammar constructs ---
indentation by a binary relation and alignment. This paper proposes a
step-by-step transformation of well-formed Adams' grammars for 
elimination of the alignment construct from the grammar. The idea that 
alignment could be avoided was suggested by Adams but no 
process for achieving this aim has been described before. 
\end{abstract}

\section{Introduction}\label{intro}

Parsing expression grammars (PEG) introduced by Ford
\cite{DBLP:conf/popl/Ford04} serve as a modern framework for specifying the
syntax of programming languages and are an alternative to the classic
context-free grammars (CFG).  The core difference between CFG and PEG is
that descriptions in CFG can be ambiguous while PEGs are inherently
deterministic.  A syntax specification written in PEG can in principle be
interpreted as a top-down parser for that syntax; in the case of left
recursion, this treatment is not straightforward but doable (see, e.g.,
\cite{DBLP:conf/sblp/MedeirosMI12}).

Formally, a PEG is a quadruple $\pgg=(\ntsn,\tmst,\pefd,\pes)$ where:
\begin{itemize}
\item $\ntsn$ is a finite set of \term{non-terminals};
\item $\tmst$ is a finite set of \term{terminals};
\item $\pefd$ 
is a function mapping each non-terminal to its replacement
(corresponding to the set of productions of CFG);
\item $\pes$ 
is the \term{start expression} (corresponding to the start symbol of CFG).
\end{itemize}
So $\pefd:\ntsn\to\Expr{\pgg}$ and $\pes\in\Expr{\pgg}$, where 
the set $\Expr{\pgg}$ of all \term{parsing expressions} writable in $\pgg$ is 
defined inductively as follows: 
\begin{enumerate}
\item $\e\in\Expr{\pgg}$ (the empty string);
\item $\tma\in\Expr{\pgg}$ for every $\tma\in\tmst$ (the terminals);
\item $\ntx\in\Expr{\pgg}$ for every $\ntx\in\ntsn$ (the non-terminals);
\item $\pep\peq\in\Expr{\pgg}$ whenever $\pep\in\Expr{\pgg}$,
$\peq\in\Expr{\pgg}$ (\term{concatenation})
\item $\pep\alt\peq\in\Expr{\pgg}$ whenever $\pep\in\Expr{\pgg}$, 
$\peq\in\Expr{\pgg}$ (\term{choice});
\item $\lkh{\pep}\in\Expr{\pgg}$ whenever $\pep\in\Expr{\pgg}$
(\term{negation}, or \term{lookahead});
\item $\rep{\pep}\in\Expr{\pgg}$ whenever $\pep\in\Expr{\pgg}$ 
(\term{repetition}).
\end{enumerate}

All constructs of PEG except for negation are direct analogues of constructs
of the EBNF form of CFG, but their semantics is always deterministic. So 
$\rep{\pep}$ repeats parsing of $\pep$ until failure, 
and $\pep\alt\peq$ always tries to parse $\pep$ first, 
$\peq$ is parsed only if $\pep$ fails. For example, the expression 
$\tma\tmb\alt\tma$ consumes the input string \texttt{ab} entirely while 
$\tma\alt\tma\tmb$ only consumes its first character. 
The corresponding EBNF expressions $\tma\tmb\mid\tma$ and $\tma\mid\tma\tmb$ 
are equivalent, both can match either $\tma$ or $\tma\tmb$ from the input
string. Negation $\lkh{\pep}$ tries to parse $\pep$ and fails if $\pep$
succeeds; if $\pep$ fails then $\lkh{\pep}$ succeeds with consuming no
input. Other constructs of EBNF like non-null repetition $\pep^+$ and
optional occurrence $[\pep]$ can be introduced to PEG as syntactic sugar.

Languages like Python and Haskell allow the syntactic structure of 
programs to be shown by indentation and alignment, instead of the more 
conventional braces and semicolons. Handling indentation and
alignment in Python has been specified in terms of extra tokens INDENT and
DEDENT that mark increasing and decreasing of indentation and 
must be generated by the lexer. In Haskell,
rules for handling indentation and alignment are more sophisticated. 
Both these languages enable to locally use a different layout mode 
where indentation does not matter, which additionally complicates the task
of formal syntax specification. 
Adams and A\uu{g}acan \cite{DBLP:conf/haskell/AdamsA14} proposed an
extension of PEG notation for specifying indentation sensitivity and 
argued that it considerably simplifies this task for Python, Haskell and
many other indentation-sensitive languages.

In this extension, expression $\ind{>}{\pep}$, for example, denotes parsing 
of $\pep$ while assuming a greater indentation than that  of the surrounding 
block. In 
general, parsing expressions may be equipped with binary relations (as was 
$>$ in the example) that must hold between the baselines of the local and the 
current indentation block. In addition, $\aln{\pep}$ denotes parsing of 
$\pep$ while assuming the first token of the input being aligned, i.e., 
positioned on the current indentation baseline. For example, 
the do expressions in Haskell can be specified by
\[
\begin{array}{lcl}
\nter{doexp}&::=&\ind{>}{\code{\kw{do}}}\csp(\nter{istmts}\alt\nter{stmts})\\
\nter{istmts}&::=&\ind{>}{(\aln{\nter{stmt}}^+)}\\
\nter{stmts}&::=&\ind{>}{\code{\kw{\{}}}\ind{\rlxd}{(\nter{stmt}(\code{\kw{;}}\nter{stmt})^*[\code{\kw{;}}]\code{\kw{\}}})}
\end{array}
\] 
Here, $\nter{istmts}$ and $\nter{stmts}$ stand for statement lists in the
indentation and relaxed mode, respectively. In the indentation mode,
a statement list is indented (marked by $>$ in the second production) and all
statements in it are aligned (marked by $\aln{\cdot}$). In the relaxed mode, 
however, relation $\rlxd$ is used to indicate that the indentation baseline
of the contents can be anything. (Technically, $\rlxd$ is the binary relation
containing all pairs of natural numbers.) Terminals $\code{\kw{do}}$ and
$\code{\kw{\{}}$ are also equipped with $>$ to meet the Haskell requirement 
that subsequent tokens of aligned blocks must be indented more than the first 
token. 

Alignment construct provides fulcra for disambiguating the often large 
variety of indentation baseline candidates. 
Besides simplicity of this grammar extension and its use, a strength of it 
lies in the fact that grammars can still serve as parsers. 

The rest of the paper is organized as follows. Section~\ref{sem} formally 
introduces additional constructs of PEG for specifying code layout, 
defines their semantics and studies their
semantic properties. In Sect.~\ref{elim}, a semantics-preserving 
process of eliminating the alignment construct from grammars is described. 
Section ~\ref{other} refers to related work and
Sect.~\ref{conc} concludes.

\section{Indentation extension of PEG}\label{sem}

Adams and A\uu{g}acan \cite{DBLP:conf/haskell/AdamsA14} extend PEGs
with the indentation and alignment constructs. We propose a
slightly different extension with three rather than two extra constructs. 
Our approach agrees with that implemented by Adams in his
\texttt{indentation} package for Haskell~\cite{implementation}, whence
calling the grammars in our approach \term{Adams' grammars} is
justified. All differences between the definitions in this paper and 
in \cite{DBLP:conf/haskell/AdamsA14} are listed and 
discussed in Subsect.~\ref{diff}.

Let $\Nat$ denote the set of all natural numbers, and let 
$\Bool=\set{\tru,\fls}$ (the Boolean domain). 
Denote by $\pws(X)$ the set of all subsets of set $X$, and let $\Rel(X)$ 
denote the set of all binary relations on set $X$, i.e., 
$\Rel(X)=\pws(X\times X)$. 
Standard examples are $>\in\Rel(\Nat)$ (consisting of all pairs $(n,m)$ of
natural numbers such that $n>m$) and $\idrel\in\Rel(\Nat)$ (the identity
relation consisting of all pairs of equal natural numbers); the indentation
extension also makes use of 
$\rlxd\in\Rel(\Nat)$ (the relation containing all pairs of natural
numbers). Whenever $\rlr\in\Rel(X)$ and $Y\subseteq X$, denote
$\rlr(Y)=\set{x\in X:\one{y\in Y}{(y,x)\in\rlr}}$ (the image of $Y$
under relation~$\rlr$). 
The inverse relation of $\rlr$ is defined by 
$\rlr^{-1}=\set{(x,y):(y,x)\in\rlr}$, and the composition of relations
$\rls$ and $\rlr$ by 
$\rls\circ\rlr=\set{(x,z):\one{y}{(x,y)\in\rls\wedge(y,z)\in\rlr}}$.
Finally, denote
$\Rel^+(X)=\set{\rlr\in\Rel(X):\all{x\in
X}{\rlr^{-1}(\set{x})\ne\0}}=\set{\rlr\in\Rel(X):\rlr(X)=X}$.

\subsection{Adams' grammars}\label{sem:adams}

Extend the definition of $\Expr{\pgg}$ given in Sect.~\ref{intro} 
with the following three additional clauses: 

\begin{enumerate}\setcounter{enumi}{7}
\item $\ind{\rlr}{\pep}\in\Expr{\pgg}$ for every $\pep\in\Expr{\pgg}$ and 
$\rlr\in\Rel(\Nat)$ (\term{indentation});
\item $\loc{\rls}{\pep}\in\Expr{\pgg}$ for every $\pep\in\Expr{\pgg}$ and
$\rls\in\Rel(\Nat)$ (\term{token position});
\item $\aln{\pep}\in\Expr{\pgg}$ for every $\pep\in\Expr{\pgg}$ 
(\term{alignment}).
\end{enumerate}

Parsing of an expression $\ind{\rlr}{\pep}$ means parsing of $\pep$ while 
assuming that the part of the input string corresponding to $\pep$ forms 
a new indentation block whose baseline is in relation $\rlr$ to the baseline 
of the surrounding block. (Baselines are identified with column numbers.)
The position construct $\loc{\rls}{\pep}$, missing in 
\cite{DBLP:conf/haskell/AdamsA14}, determines how tokens of the input 
can be situated \wrt the current indentation baseline. 
Finally, parsing an expression $\aln{\pep}$ means parsing of~$\pep$ while 
assuming the first token of the input being positioned on the current 
indentation baseline 
(unlike the position operator, this construct
does not affect processing the subsequent tokens).

Inspired by the \code{indentation} 
package~\cite{implementation}, we call the relations that determine
token positioning \wrt the indentation baseline \term{token modes}.
In the token mode $>$ for example, tokens may appear only to the right of
the indentation baseline.
Applying the position operator with relation~$>$ to parts of 
Haskell grammar to be parsed in the indentation mode avoids 
indenting every single terminal in the
example in Sect.~\ref{intro}. Also, 
indenting terminals with~$>$
is inadequate for do expressions occurring inside a block of relaxed
mode but the position construct can be easily used to change the token mode 
for such blocks (e.g., to $\ge$).

We call a PEG extended with these three constructs a \npeg. Recall from
Sect.~\ref{intro} that $\ntsn$ and $\tmst$ denote the set of non-terminal
and terminal symbols of the grammar, respectively, and
$\pefd:\ntsn\to\Expr{\pgg}$ is the production function.
Concerning the semantics of \npeg, each expression
parses an input string of terminals ($w\in\tmst^*$) in the context of a current 
set of indentation baseline candidates ($I\in\pws(\Nat)$) and a current 
alignment flag indicating 
whether the next terminal should be aligned or not ($b\in\Bool$), assuming a 
certain token mode
($\rlt\in\Rel(\Nat)$). Parsing may succeed, fail, or diverge. If parsing
succeeds, it returns as a result a new triple containing the rest of the 
input $w'$, a new set $I'$ of baseline candidates updated according to the 
information gathered during parsing, 
and a new alignment 
flag $b'$. This result is denoted by $\top(w',I',b')$. If parsing fails, 
there is no result in a triple form; failure is denoted by $\bot$. 

Triples of the form $(w,I,b)\in\tmst^*\times\pws(\Nat)\times\Bool$ 
are behaving as \term{operation states} of parsing, as each parsing step 
may use these data and update them. We will write 
$\State=\tmst^*\times\pws(\Nat)\times\Bool$ (as we never deal with different
terminal sets, dependence on $\tmst$ is not explicitly marked), and denote by 
$\State+1$ the set of possible results of parsing, i.e., 
$\set{\top(\sts):\sts\in\State}\un\set{\bot}$.

The assertion that parsing expression $\pee$ in grammar $\pgg$ with input 
string $w$ in the context of $I$ and $b$ assuming token mode $\rlt$ 
results in 
$o\in\State+1$ is denoted by $\parse{\pgg}{\pee}{\rlt}{(w,I,b)}{o}$. 
The formal definition below must be interpreted inductively, i.e., 
an assertion of the form $\parse{\pee}{\pgg}{\rlt}{\sts}{\steo}$ is valid iff 
it has a finite derivation by the following ten rules: 

\begin{enumerate}
\item $\parse{\pgg}{\e}{\rlt}{\sts}{\top(\sts)}$.
\item For every $\tma\in\tmst$, $\parse{\pgg}{\tma}{\rlt}{(w,I,b)}{\steo}$ 
holds in two cases:
\begin{itemize}
\item If $\steo=\top(w',I',\fls)$ for $w'$, $I'$, $i$ such that
$w=\tma^iw'$ ($\tma^i$ denotes
$\tma$ occurring at column~$i$) and either $b=\fls$ and 
$i\in\rlt^{-1}(I)$, $I'=I\nu\rlt(\set{i})$, or $b=\tru$ and $i\in I$, 
$I'=\set{i}$;
\item If $\steo=\bot$, and there are no $w'$ and $i$ such that $w=\tma^iw'$ 
with either $b=\fls$ and $i\in\rlt^{-1}(I)$ or $b=\tru$ and $i\in I$.
\end{itemize}
\item For every $\ntx\in\ntsn$, $\parse{\pgg}{\ntx}{\rlt}{\sts}{\steo}$ holds 
if $\parse{\pgg}{\pefd(\ntx)}{\rlt}{\sts}{\steo}$ holds.
\item For every $\pep,\peq\in\Expr{\pgg}$, 
$\parse{\pgg}{\pep\peq}{\rlt}{\sts}{\steo}$ holds in two cases:
\begin{itemize}
\item If there exists a triple $\sts'$ such that 
$\parse{\pgg}{\pep}{\rlt}{\sts}{\top(\sts')}$ and 
$\parse{\pgg}{\peq}{\rlt}{\sts'}{\steo}$;
\item If $\parse{\pgg}{\pep}{\rlt}{\sts}{\bot}$ and $\steo=\bot$.
\end{itemize}
\item For every $\pep,\peq\in\Expr{\pgg}$, 
$\parse{\pgg}{\pep\alt\peq}{\rlt}{\sts}{\steo}$ holds in two cases:
\begin{itemize}
\item If there exists a triple $\sts'$ such that 
$\parse{\pgg}{\pep}{\rlt}{\sts}{\top(\sts')}$ and $\steo=\top(\sts')$;
\item If $\parse{\pgg}{\pep}{\rlt}{\sts}{\bot}$ and 
$\parse{\pgg}{\peq}{\rlt}{\sts}{\steo}$.
\end{itemize}
\item For every $\pep\in\Expr{\pgg}$,
$\parse{\pgg}{\lkh{\pep}}{\rlt}{\sts}{\steo}$ holds in two cases:
\begin{itemize}
\item If $\parse{\pgg}{\pep}{\rlt}{\sts}{\bot}$ and $\steo=\top(\sts)$;
\item If there exists a triple $\sts'$ such that 
$\parse{\pgg}{\pep}{\rlt}{\sts}{\top(\sts')}$ and $\steo=\bot$.
\end{itemize}
\item For every $\pep\in\Expr{\pgg}$,
$\parse{\pgg}{\rep{\pep}}{\rlt}{\sts}{\steo}$ holds in two cases:
\begin{itemize}
\item If $\parse{\pgg}{\pep}{\rlt}{\sts}{\bot}$ and $\steo=\top(\sts)$;
\item If there exists a triple $\sts'$ such that
$\parse{\pgg}{\pep}{\rlt}{\sts}{\top(\sts')}$ and
$\parse{\pgg}{\rep{\pep}}{\rlt}{\sts'}{\steo}$.
\end{itemize}
\item For every $\pep\in\Expr{\pgg}$ and $\rlr\in\Rel(\Nat)$,
$\parse{\pgg}{\ind{\rlr}{\pep}}{\rlt}{(w,I,b)}{\steo}$ holds in two cases:
\begin{itemize}
\item If there exists a triple $(w',I',b')$ such that
$\parse{\pgg}{\pep}{\rlt}{(w,\rlr^{-1}(I),b)}{\top(w',I',b')}$ and
$\steo=\top(w',I\nu\rlr(I'),b')$;
\item If $\parse{\pgg}{\pep}{\rlt}{(w,\rlr^{-1}(I),b)}{\bot}$ and
$\steo=\bot$.
\end{itemize}
\item For every $\pep\in\Expr{\pgg}$ and $\rls\in\Rel(\Nat)$,
$\parse{\pgg}{\loc{\rls}{\pep}}{\rlt}{\sts}{\steo}$ holds if 
$\parse{\pgg}{\pep}{\rls}{\sts}{\steo}$ holds. 
\item For every $\pep\in\Expr{\pgg}$,
$\parse{\pgg}{\aln{\pep}}{\rlt}{(w,I,b)}{\steo}$ holds in two cases:
\begin{itemize}
\item If there exists a triple $(w',I',b')$ such that
$\parse{\pgg}{\pep}{\rlt}{(w,I,\tru)}{\top(w',I',b')}$ and
$\steo=\top(w',I',b\wedge b')$;
\item If $\parse{\pgg}{\pep}{\rlt}{(w,I,\tru)}{\bot}$ and $\steo=\bot$.
\end{itemize}
\end{enumerate}
The idea behind the conditions $i\in\rlt^{-1}(I)$ and $i\in I$ occurring in 
clause 2 is that any column~$i$ where a token may appear is in 
relation~$\rlt$ with the current indentation baseline (known to be in $I$) 
if no alignment flag is set, and coincide with the indentation 
baseline otherwise. For the same reason, consuming a token in column $i$ 
restricts the set of allowed indentations to $\rlt(\set{i})$ or $\set{i}$ 
depending on the alignment flag.
In both cases, the alignment flag is set to $\fls$. 
In clause 8 for $\ind{\rlr}{\pep}$, the set $I$ of allowed indentation is 
replaced by 
$\rlr^{-1}(I)$ as the local indentation baseline must be in relation $\rlr$
with the current indentation baseline known to be in $I$. After successful
parsing of $\pep$ with the resulting set of allowed local
indentations being $I'$, the set of allowed indentations of the surrounding block 
is restricted to $\rlr(I')$. 
Clause 10 similarly 
operates on the alignment flag.

For a toy example, consider parsing of $\ind{>}{\aln{\tma\tmb}}$ with the
operation state $(\tma^2\tmb^3,\Nat,\fls)$ assuming the token mode $\ge$. 
For that, we must parse 
$\aln{\tma\tmb}$ with $(\tma^2\tmb^3,\Nat\sm\set{0},\fls)$ by clause 8 since
$>^{-1}(\Nat)=\Nat\sm\set{0}$. For that in turn, we must parse $\tma\tmb$ 
with $(\tma^2\tmb^3,\Nat\sm\set{0},\tru)$ by clause 10. By clause 
2, we have
$\parse{\pgg}{\tma}{\ge\,}{\!(\tma^2\tmb^3,\Nat\sm\set{0},\tru)}{\top(\tmb^3,\set{2},\fls)}$ 
(as $2\in\Nat\sm\set{0}$)
and $\parse{\pgg}{\tmb}{\ge\,}{\!(\tmb^3,\set{2},\fls)}{\top(\e,\set{2},\fls)}$ 
(as $(2,3)\in\,\ge^{-1}$). Therefore, by clause 4, 
$\parse{\pgg}{\tma\tmb}{\ge\,}{(\tma^2\tmb^3,\Nat\sm\set{0},\tru)}{\top(\e,\set{2},\fls)}$.
Finally, 
$\parse{\pgg}{\aln{\tma\tmb}}{\ge\,}{(\tma^2\tmb^3,\Nat\sm\set{0},\fls)}{\top(\e,\set{2},\fls)}$ 
and
$\parse{\pgg}{\ind{>}{\aln{\tma\tmb}}}{\ge\,}{(\tma^2\tmb^3,\Nat,\fls)}{\top(\e,\set{0,1},\fls)}$
by clauses 10 and 8. The set $\set{0,1}$ in the final state shows that only
$0$ and $1$ are still candidates for the indentation baseline outside the
parsed part of the input (before parsing, the candidate set was the whole $\Nat$).

Note that this definition involves circular dependencies. For instance, if
$\pefd(\ntx)=\ntx$ for some $\ntx\in\ntsn$ then 
$\parse{\pgg}{\ntx}{\rlt}{\sts}{\steo}$ if 
$\parse{\pgg}{\ntx}{\rlt}{\sts}{\steo}$ by clause 3. 
There is no result of parsing in such cases (not
even $\bot$). We call this behaviour \term{divergence}.

\subsection{Properties of the semantics}\label{sem:prop}

Ford \cite{DBLP:conf/popl/Ford04} proves that parsing in PEG is unambiguous, 
whereby the consumed part of an input string always is its prefix.
Theorem~\ref{sem:uniq} below is an analogous result for \npeg. Besides 
the uniqueness of the result of parsing, it states that 
if we only consider relations in $\Rel^+(\Nat)$ then the whole operation 
state in our setting is in a certain sense decreasing during parsing. 

Denote by $\ge$ the \term{suffix order} of strings (i.e., $w\ge w'$ iff
$w=uw'$ for some $u\in\tmst^*$) and by $\sqsupseteq$ the 
\term{implication order} of truth values
(i.e., $\tru\sqsupset\fls$). Denote by $\geqslant$ the pointwise order 
on operation states, i.e., $(w,I,b)\geqslant(w',I',b')$ iff $w\ge w'$,
$I\supseteq I'$ and $b\sqsupseteq b'$. 

\begin{theorem}\label{sem:uniq}
Let $\pgg=(\ntsn,\tmst,\pefd,\pes)$ be a \npeg, $\pee\in\Expr{\pgg}$, 
$\rlt\in\Rel^+(\Nat)$ and
$\sts\in\State$. Then 
$\parse{\pgg}{\pee}{\rlt}{\sts}{\steo}$ for at most one $\steo$, whereby 
$\steo=\top(\sts')$ implies $\sts\geqslant\sts'$. Also if
$\sts=(w,I,b)$ and $\sts'=(w',I',b')$ then $\sts\ne\sts'$ implies both 
$w>w'$ and $b'=\fls$, and $I\ne\0$ implies $I'\ne\0$.
\end{theorem}

\begin{proof}
By induction on the shape of the derivation tree of the assertion
$\parse{\pgg}{\pee}{\rlt}{\sts}{\steo}$.
\end{proof}

Theorem~\ref{sem:uniq} enables to observe a common pattern in the semantics 
of indentation and alignment. Denoting by $\pefk(\pep)$ either 
$\ind{\rlr}{\pep}$ or $\aln{\pep}$, both clauses 8 and 10 have the following 
form, parametrized on two mappings $\a,\g:\State\to\State$:

\begin{enumerate}
\item[]
For $\pep\in\Expr{\pgg}$,
$\parse{\pgg}{\pefk(\pep)}{\rlt}{\sts}{\steo}$ holds in two cases:
\begin{itemize}
\item If there exists a state $\sts'$ such that
$\parse{\pgg}{\pep}{\rlt}{\a(\sts)}{\top(\sts')}$ and
$\steo=\top(\sts\wedge\g(\sts'))$;
\item If $\parse{\pgg}{\pep}{\rlt}{\a(\sts)}{\bot}$ and $\steo=\bot$.
\end{itemize}
\end{enumerate}

The meanings of indentation and alignment constructs are distinguished solely
by $\a$ and $\g$. For many properties, proofs that rely on this 
abstract common definition can be carried out, assuming that 
$\g$ is monotone, preserves the largest element and follows together with 
$\a$ the axiom $x\wedge\g(y)\le\g(\a(x)\wedge y)$.
The class of all meet semilattices $L$ with top element, equipped 
with mappings $\a$, $\g$ satisfying these three conditions, 
contains identities (i.e., semilattices $L$ with $\a=\g=\id_L$) and is closed 
under compositions (of different $\a$, $\g$ defined on the same semilattice
$L$) and under direct products. 
If $\rlr\in\Rel^+(\Nat)$ then the conditions hold for
$\a_1,\g_1:\pws(\Nat)\to\pws(\Nat)$ with $\a_1(I)=\rlr^{-1}(I)$, 
$\g_1(I)=\rlr(I)$, 
similarly in the case if $\a_2,\g_2:\Bool\to\Bool$ with $\a_2(b)=\tru$,
$\g_2(b)=b$. Now the direct product of the identities of $\tmst^*$ and
$\Bool$ with $(\a_1,\g_1)$ on $\pws(\Nat)$ gives the indentation case, 
and the direct 
product of the identities of $\tmst^*$ and $\pws(\Nat)$ and the Boolean 
lattice $\Bool$ with $(\a_2,\g_2)$ gives the alignment case.

If $\a,\g$ satisfy the conditions then $\g(\a(x))\ge x$ since
$x=x\wedge\top=x\wedge\g(\top)\le\g(\a(x)\wedge\top)=\g(\a(x))$. 
Adding dual conditions ($\a$ monotone, $\a(\bot)=\bot$ and
$\a(x)\vee y\ge\a(x\vee\g(y))$) would make $(\a,\g)$ a
Galois' connection. In our cases, the dual axioms do not hold.

\subsection{Semantic equivalence}\label{sem:equiv}

\begin{definition}
Let $\pgg=(\ntsn,\tmst,\pefd,\pes)$ be a \npeg\ and $\pep,\peq\in\Expr{\pgg}$. 
We say that $\pep$ and $\peq$ are \term{semantically equivalent} in $\pgg$ 
and denote $\pep\sim_{\pgg}\peq$ iff 
$\parse{\pgg}{\pep}{\rlt}{\sts}{\steo}\iff\parse{\pgg}{\peq}{\rlt}{\sts}{\steo}$ 
for every $\rlt\in\Rel^+(\Nat)$, 
$\sts\in\State$ and $\steo\in\State+1$.
\end{definition}

For example, one
can easily prove that $\pep\e\sim_{\pgg}\pep\sim_{\pgg}\e\pep$, 
$\pep(\peq\per)\sim_{\pgg}(\pep\peq)\per$, 
$\pep\alt(\peq\alt\per)\sim_{\pgg}(\pep\alt\peq)\alt\per$, 
$\pep(\peq\alt\per)\sim_{\pgg}\pep\peq\alt\pep\per$,
$\pep\alt\peq\sim_{\pgg}\pep\alt\lkh{\pep}\peq$ for all
$\pep,\peq,\per\in\Expr{\pgg}$ \cite{DBLP:conf/popl/Ford04}. 
We are particularly
interested in equivalences involving the additional
operators of \npeg. In Sect.~\ref{elim}, they will be useful in 
eliminating alignment and position operators. The following 
Theorem~\ref{sem:maineq} states distributivity laws of the three 
new operators of \npeg\ \wrt other constructs:

\begin{theorem}\label{sem:maineq}
Let $\pgg=(\ntsn,\tmst,\pefd,\pes)$ be a \npeg. Then:
\begin{enumerate}
\item \label{sem:maineq:loc}
$\loc{\rls}{\e}\sim_{\pgg}\e$,
$\loc{\rls}{(\pep\peq)}\sim_{\pgg}\loc{\rls}{\pep}\loc{\rls}{\peq}$,
$\loc{\rls}{(\pep\alt\peq)}\sim_{\pgg}\loc{\rls}{\pep}\alt\loc{\rls}{\peq}$,
$\loc{\rls}{(\lkh{\pep})}\sim_{\pgg}\lkh{\loc{\rls}{\pep}}$,
$\loc{\rls}{(\rep{\pep})}\sim_{\pgg}\rep{(\loc{\rls}{\pep})}$,
$\loc{\rls}{(\ind{\rlr}{\pep})}\sim_{\pgg}\ind{\rlr}{(\loc{\rls}{\pep})}$,
$\loc{\rls}{\aln{\pep}}\sim_{\pgg}\aln{\loc{\rls}{\pep}}$ for all
$\rls\in\Rel^+(\Nat)$;
\item $\ind{\rlr}{\e}\sim_{\pgg}\e$,
$\ind{\rlr}{(\pep\alt\peq)}\sim_{\pgg}\ind{\rlr}{\pep}\alt\ind{\rlr}{\peq}$,
$\ind{\rlr}{(\lkh{\pep})}\sim_{\pgg}\lkh\ind{\rlr}{\pep}$, 
$\ind{\rlr}{(\loc{\rls}{\pep})}\sim_{\pgg}\loc{\rls}{(\ind{\rlr}{\pep})}$ 
for all $\rlr\in\Rel^+(\Nat)$;
\item $\aln{\e}\sim_{\pgg}\e$,
$\aln{\pep\alt\peq}\sim_{\pgg}\aln{\pep}\alt\aln{\peq}$,
$\aln{\lkh{\pep}}\sim_{\pgg}\lkh\aln{\pep}$, 
$\aln{\loc{\rls}{\pep}}\sim_{\pgg}\loc{\rls}{\aln{\pep}}$.
\end{enumerate}
\end{theorem}

\begin{proof}
The equivalences in claim 1 hold as
the token mode steadily distributes to each case of the
semantics definition. Those in claims 2 and 3 have straightforward 
proofs using the joint form of the semantics of indentation and alignment 
and the axioms of $\a$, $\g$.
\end{proof}

Note that indentation does not distribute with concatenation, i.e.,
$\ind{\rlr}{(\pep\peq)}\nsim_{\pgg}\ind{\rlr}{\pep}\ind{\rlr}{\peq}$. This is
because $\ind{\rlr}{(\pep\peq)}$ assumes one indentation block with a
baseline common to $\pep$ and $\peq$ while $\ind{\rlr}{\pep}\ind{\rlr}{\peq}$ 
tolerates different baselines for $\pep$ and $\peq$. For example, take 
$\pep=\tma\in\tmst$, $\peq=\tmb\in\tmst$, let the token mode be
$\idrel$ and the input state be $(\tma^1\tmb^2,\Nat,\fls)$ (recall that
$\tma^i$ means terminal $\tma$ occurring in column $i$). We have
$\parse{\pgg}{\tma}{\idrel}{(\tma^1\tmb^2,\Nat\sm\set{0},\fls)}{\top(\tmb^2,\set{1},\fls)}$ 
and $\parse{\pgg}{\tmb}{\idrel}{(\tmb^2,\set{1},\fls)}{\bot}$ (since
$(2,1)\nin\idrel$), therefore
$\parse{\pgg}{\tma\tmb}{\idrel}{(\tma^1\tmb^2,\Nat\sm\set{0},\fls)}{\bot}$ and 
$\parse{\pgg}{\ind{>}{(\tma\tmb)}}{\idrel}{(\tma^1\tmb^2,\Nat,\fls)}{\bot}$. 
On the other hand,
$\parse{\pgg}{\tma}{\idrel}{(\tma^1\tmb^2,\Nat\sm\set{0},\fls)}{\top(\tmb^2,\set{1},\fls)}$
implies
$\parse{\pgg}{\ind{>}{\tma}}{\idrel}{(\tma^1\tmb^2,\Nat,\fls)}{\top(\tmb^2,\set{0},\fls)}$
(since $\Nat\nu(>(\set{1}))=\set{0}$) and, analogously, 
$\parse{\pgg}{\ind{>}{\tmb}}{\idrel}{(\tmb^2,\set{0},\fls)}{\top(\e,\set{0},\fls)}$ 
(since $>^{-1}(\set{0})=\Nat\sm\set{0}\ni 2$ and
$\set{0}\nu(>(\set{2}))=\set{0}$). Consequently, 
$\parse{\pgg}{\ind{>}{\tma}\ind{>}{\tmb}}{\idrel}{(\tma^1\tmb^2,\Nat,\fls)}{\top(\e,\set{0},\fls)}$.

We can however prove the following facts:

\begin{theorem}\label{sem:absorb}
Let $\pgg=(\ntsn,\tmst,\pefd,\pes)$ be a \npeg.
\begin{enumerate}
\item Identity indentation law: For all $\pep\in\Expr{\pgg}$, $\ind{\idrel}{\pep}\sim_{\pgg}\pep$.
\item Composition law of indentations: For all $\pep\in\Expr{\pgg}$ and $\rlr,\rls\in\Rel^+(\Nat)$, 
$\ind{\rls}{(\ind{\rlr}{\pep})}\sim_{\pgg}\ind{\rls\circ\rlr}{\pep}$. 
\item Distributivity of indentation and alignment: For all $\pep\in\Expr{\pgg}$ and $\rlr\in\Rel^+(\Nat)$, 
$\ind{\rlr}{\aln{\pep}}\sim_{\pgg}\aln{\ind{\rlr}{\pep}}$. 
\item Idempotence of alignment: For all $\pep\in\Expr{\pgg}$, $\aln{\aln{\pep}}\sim_{\pgg}\aln{\pep}$. 
\item \label{sem:absorb:loc}
Cancellation of outer token modes: 
For all $\pep\in\Expr{\pgg}$ and $\rls,\rlt\in\Rel(\Nat)$, 
$\loc{\rlt}{(\loc{\rls}{\pep})}\sim_{\pgg}\loc{\rls}{\pep}$. 
\item Terminal alignment property: For all $\tma\in\tmst$,
$\aln{\tma}\sim_{\pgg}\loc{\idrel}{\tma}$. 
\end{enumerate}
\end{theorem}

\begin{proof}
For claim 1, note that an indentation with the identity relation 
$\idrel$ corresponds to $\a$, $\g$ being identity mappings. Hence
\[\renewcommand{\arraystretch}{1.6}
\begin{array}{lcl}
\parse{\pgg}{\ind{\idrel}{\pep}}{\rlt}{\sts}{\steo}
&\iff&\bigor
{\one{\sts'}{\parse{\pgg}{\pep}{\rlt}{\sts}{\top(\sts')}\la\steo=\top(\sts\wedge\sts')}}
{\parse{\pgg}{\pep}{\rlt}{\sts}{\bot}\wedge\steo=\bot}\\
&\iff&\bigor
{\one{\sts'}{\parse{\pgg}{\pep}{\rlt}{\sts}{\top(\sts')}\la\steo=\top(\sts')}}
{\parse{\pgg}{\pep}{\rlt}{\sts}{\bot}\wedge\steo=\bot}\\
&\iff&\parse{\pgg}{\pep}{\rlt}{\sts}{\steo}\mbox{,}
\end{array}
\]
where $\sts\wedge\sts'$ can be replaced with $\sts'$ because
$\sts\geqslant\sts'$ by Theorem~\ref{sem:uniq}.

Concerning claims 2--4, let $\pefk_1,\pefk_2$ be two constructs whose
semantics follow the common pattern of indentation and alignment with
mapping pairs $(\a_1,\g_1)$ and $(\a_2,\g_2)$, respectively. Then
\[\renewcommand{\arraystretch}{1.6}
\begin{array}{lcl}
\multicolumn{3}{@{}l}{\parse{\pgg}{\pefk_2(\pefk_1(\pep))}{\rlt}{\sts}{\steo}}\\
\qquad&\iff&\bigor
{\one{\sts'}{\parse{\pgg}{\pefk_1(\pep)}{\rlt}{\a_2(\sts)}{\top(\sts')}\la\steo=\top(\sts\wedge\g_2(\sts'))}}
{\parse{\pgg}{\pefk_1(\pep)}{\rlt}{\a_2(\sts)}{\bot}\la\steo=\bot}\\
&\iff&\bigor
{\one{\sts''}{\parse{\pgg}{\pep}{\rlt}{\a_1(\a_2(\sts))}{\top(\sts'')}\la\steo=\top(\sts\wedge\g_2(\a_2(\sts)\wedge\g_1(\sts'')))}}
{\parse{\pgg}{\pep}{\rlt}{\a_1(\a_2(\sts))}{\bot}\la\steo=\bot}\mbox{.}
\end{array}
\]
By monotonicity of $\g_2$ and the fact that $\sts\leqslant\g_2(\a_2(\sts))$, we 
have 
$\sts\wedge\g_2(\a_2(\sts)\wedge\g_1(\sts''))\leqslant\sts\wedge\g_2(\a_2(\sts))\wedge\g_2(\g_1(\sts''))=\sts\wedge\g_2(\g_1(\sts''))$. 
By the third axiom of $\a_2$ and $\g_2$, we also have 
$\g_2(\a_2(\sts)\wedge\g_1(\sts''))\geqslant\sts\wedge\g_2(\g_1(\sts''))$
whence
$\sts\wedge\g_2(\a_2(\sts)\wedge\g_1(\sts''))\geqslant\sts\wedge\g_2(\g_1(\sts''))$. 
Consequently, 
$\sts\wedge\g_2(\a_2(\sts)\wedge\g_1(\sts''))$ 
can be replaced with $\sts\wedge\g_2(\g_1(\sts''))$. Hence 
the semantics of the composition of 
$\pefk_1$ and $\pefk_2$ follows the pattern of semantics of indentation
and alignment for mappings 
$\a_1\circ\a_2$ and $\g_2\circ\g_1$. To prove claim 2, it now suffices to 
observe that the mappings $\a,\g$ in the semantics of 
$\ind{\rls\circ\rlr}{(\cdot)}$ equal the compositions of the corresponding
mappings for the semantics of $\ind{\rlr}{(\cdot)}$ and
$\ind{\rls}{(\cdot)}$. For claim 3, it suffices to observe that the 
mappings $\a$, $\g$ given for an indentation and for alignment modify
different parts of the operation state whence their order of application is
irrelevant. 
Claim 4 holds because the mappings $\a$, $\g$ in the alignment semantics are
both idempotent.

Finally, claim 5 is trivial and claim 6 follows from a straightforward case 
study.
\end{proof}

Theorems \ref{sem:maineq} and \ref{sem:absorb} enact bringing alignments 
through all syntactic constructs except concatenation. 
Alignment does not distribute with 
concatenation, because in parsing of an expression of the form 
$\aln{\pep\peq}$, the terminal to be aligned can be in 
the part of the input consumed by $\pep$ or (if parsing of $\pep$ succeeds 
with consuming no input) by $\peq$. Alignment can nevertheless be moved
through concatenation 
if any successful parsing of the first expression in the 
concatenation either never consumes any input or always consumes some input:

\begin{theorem}\label{sem:alnhelp}
Let $\pgg=(\ntsn,\tmst,\pefd,\pes)$ be a \npeg\ and
$\pep,\peq\in\Expr{\pgg}$. 
\begin{enumerate}
\item 
If $\parse{\pgg}{\pep}{\rlt}{\sts}{\top(\sts')}$ implies $\sts'=\sts$
for all $\rlt\in\Rel^+(\Nat)$, $\sts,\sts'\in\State$, then
$\aln{\pep\peq}\sim_{\pgg}\aln{\pep}\aln{\peq}$.
\item If $\parse{\pgg}{\pep}{\rlt}{\sts}{\top(\sts')}$ implies $\sts'\ne\sts$
for all $\rlt\in\Rel^+(\Nat)$, $\sts,\sts'\in\State$, then
$\aln{\pep\peq}\sim_{\pgg}\aln{\pep}\peq$.
\end{enumerate}
\end{theorem}

\begin{proof}
Straightforward case study.
\end{proof}

Theorem~\ref{sem:alnhelp} (1) holds also for indentation 
(instead of alignment), the same proof in terms of $\a$, $\g$ is valid.
Finally, the following theorem states that position and indentation of 
terminals are equivalent if the alignment flag is false and 
the token mode is the identity:

\begin{theorem}\label{sem:locind}
Let $\pgg=(\ntsn,\tmst,\pefd,\pes)$ be a \npeg. Let $\tma\in\tmst$, 
$\rls\in\Rel^+(\Nat)$ and $w\in\tmst^*$, $I\in\pws(\Nat)$,
$\steo\in\State+1$. 
Then
$\parse{\pgg}{\loc{\rls}{\tma}}{\idrel}{(w,I,\fls)}{\steo}\iff\parse{\pgg}{\ind{\rls}{\tma}}{\idrel}{(w,I,\fls)}{\steo}$.
\end{theorem}

\begin{proof}
Straightforward case study.
\end{proof}

\subsection{Differences of our approach from previous work}\label{diff}

Our specification of \npeg differs from the definition 
used by Adams and A\uu{g}acan \cite{DBLP:conf/haskell/AdamsA14} by three
essential aspects listed below.  
The last two discrepancies can be understood as
bugs in the original description that have been corrected in the Haskell
\texttt{indentation} package by Adams \cite{implementation}.  This package
also provides means for locally changing the token mode. 
All in all, our modifications fully agree with the 
\texttt{indentation} package.

\begin{enumerate}
\item The position operator $\loc{\rls}{\pep}$ is missing in
\cite{DBLP:conf/haskell/AdamsA14}. The treatment there assumes just one 
default token mode applying to the whole grammar, whence token
positions deviating from the default must be specified using the 
indentation operator. The benefits of the position operator were shortly
discussed in Subsect.~\ref{sem:adams}.

\item According to the grammar semantics provided in
\cite{DBLP:conf/haskell/AdamsA14}, the alignment flag is never 
changed at the end of parsing of an expression of the form $\aln{\pep}$. 
This is not appropriate if $\pep$ succeeds without consuming any token,
as the alignment flag would unexpectedly remain true during 
parsing of the next token that is out of scope of the alignment
operator. The value the alignment flag had before 
starting parsing $\aln{\pep}$ should be restored in this case. 
This is the purpose of conjunction in the
alignment semantics described in this paper.

\item In \cite{DBLP:conf/haskell/AdamsA14}, 
an alignment is interpreted \wrt the indentation 
baseline of the block that corresponds to the parsing expression to 
which the alignment operator is applied. Indentation operators occurring
inside this expression and 
processed while the alignment flag is true are neglected. 
In the semantics described in our paper, raising the alignment flag
does not suppress new indentations. Alignments are 
interpreted \wrt the indentation baseline in force at the aligned token 
site. This seems more appropriate than the former approach where 
the indentations cancelled because of an alignment do not apply even to 
the subsequent non-aligned tokens. Distributivity of
indentation and alignment fails in the semantics of 
\cite{DBLP:conf/haskell/AdamsA14}. 
Note that alignment of a block nevertheless 
suppresses the influence of position operators whose scope extend over
the first token of the block.

\end{enumerate}

Our grammar semantics has also two purely formal deviations from the
semantics used by Adams and
A\uu{g}acan~\cite{DBLP:conf/haskell/AdamsA14} and
Ford~\cite{DBLP:conf/popl/Ford04}.

\begin{enumerate}
\item 
We keep track of the rest of the input in the operation state while 
both \cite{DBLP:conf/haskell/AdamsA14,DBLP:conf/popl/Ford04} 
expose the consumed part of the input instead. 
This difference was introduced for simplicity and to 
achieve uniform decreasing of operation states in Theorem~\ref{sem:uniq}. 

\item We do not have explicit step counts. They were used in 
\cite{DBLP:conf/popl/Ford04} to compose proofs by induction. We provide 
analogous proofs by induction on the shape of derivation trees. 

\end{enumerate}

\section{Elimination of alignment and position operators}\label{elim}

Adams \cite{DBLP:conf/popl/Adams13} describes alignment elimination in the
context of CFGs.  In \cite{DBLP:conf/haskell/AdamsA14}, Adams and
A\uu{g}acan claim that alignment elimination process for PEGs is more
difficult due to the lookahead construct.  To our knowledge, no concrete
process of semantics-preserving alignment elimination is described for PEGs
before.  We provide one below for well-formed grammars. 
We rely on the existence of
position operators in the grammar; this is not an issue since
we also show that position operators can be eliminated from alignment-free
grammars.

\subsection{Approximation semantics and well-formed
expressions}\label{trans:wf}

For defining well-formedness, we first need to introduce 
\term{approximation semantics} that consists
of assertions of the form $\appr{\pgg}{\pee}{n}$ where $\pee\in\Expr{\pgg}$
and $n\in\set{-1,0,1}$. 
This semantics is a decidable extension of the predicate that 
tells whether parsing of $\pee$ may succeed with consuming no input 
(result $0$), succeed with consuming some input (result $1$) or fail 
(result $-1$). No particular input strings, indentation sets etc.
are involved, whence the semantics is not deterministic. The following set
of clauses define the approximation semantics inductively.

\begin{enumerate}
\item $\appr{\pgg}{\e}{0}$.
\item For every $\tma\in\tmst$, $\appr{\pgg}{\tma}{1}$ and $\appr{\pgg}{\tma}{-1}$.
\item For every $\ntx\in\ntsn$, 
$\appr{\pgg}{\ntx}{n}$ if $\appr{\pgg}{\pefd(\ntx)}{n}$.
\item For every $\pep,\peq\in\Expr{\pgg}$, $\appr{\pgg}{\pep\peq}{n}$ holds
in four cases:
\begin{itemize}
\item $\appr{\pgg}{\pep}{0}$, $\appr{\pgg}{\peq}{0}$ and $n=0$;
\item There exist $n',n''\in\set{0,1}$ such that 
$\appr{\pgg}{\pep}{n'}$, $\appr{\pgg}{\peq}{n''}$,
$1\in\set{n',n''}$ and $n=1$;
\item There exists $n'\in\set{0,1}$ such that 
$\appr{\pgg}{\pep}{n'}$, $\appr{\pgg}{\peq}{-1}$ and
$n=-1$;
\item $\appr{\pgg}{\pep}{-1}$ and $n=-1$.
\end{itemize}
\item For every $\pep,\peq\in\Expr{\pgg}$, $\appr{\pgg}{\pep\alt\peq}{n}$
holds in two cases:
\begin{itemize}
\item $\appr{\pgg}{\pep}{n}$ and $n\in\set{0,1}$;
\item $\appr{\pgg}{\pep}{-1}$ and $\appr{\pgg}{\peq}{n}$.
\end{itemize}
\item For every $\pep\in\Expr{\pgg}$, $\appr{\pgg}{\lkh{\pep}}{n}$ holds in
two cases:
\begin{itemize}
\item $\appr{\pgg}{\pep}{-1}$ and $n=0$;
\item There exists $n'\in\set{0,1}$ such that $\appr{\pgg}{\pep}{n'}$ and $n=-1$.
\end{itemize}
\item For every $\pep\in\Expr{\pgg}$, $\appr{\pgg}{\rep{\pep}}{n}$ holds in
two cases:
\begin{itemize}
\item $\appr{\pgg}{\pep}{-1}$ and $n=0$;
\item $\appr{\pgg}{\pep}{-1}$, $\appr{\pgg}{\pep}{1}$ and $n=1$.
\end{itemize}
\item For every $\pep\in\Expr{\pgg}$ and $\rlr\in\Rel(\Nat)$, 
$\appr{\pgg}{\ind{\rlr}{\pep}}{n}$ if $\appr{\pgg}{\pep}{n}$.
\item For every $\pep\in\Expr{\pgg}$ and $\rls\in\Rel(\Nat)$, 
$\appr{\pgg}{\loc{\rls}{\pep}}{n}$ if $\appr{\pgg}{\pep}{n}$.
\item For every $\pep\in\Expr{\pgg}$, 
$\appr{\pgg}{\aln{\pep}}{n}$ if $\appr{\pgg}{\pep}{n}$.
\end{enumerate}

On the PEG constructs (1--7), our definition basically copies that given by 
Ford~\cite{DBLP:conf/popl/Ford04}, except for the case
$\appr{\pgg}{\rep{\pep}}{1}$ where our definition requires 
$\appr{\pgg}{\pep}{-1}$ besides $\appr{\pgg}{\pep}{1}$. This is sound since 
if parsing of $\pep$ never fails then 
parsing of $\rep{\pep}$ cannot terminate. The difference does not matter in
the grammar transformations below as they assume repetition-free grammars.

\begin{theorem}
Let $\pgg=(\ntsn,\tmst,\pefd,\pes)$ be a \npeg. Assume that
$\parse{\pgg}{\pee}{\rlt}{\sts}{\steo}$ for some $\rlt\in\Rel(\Nat)$ and 
$\sts\in\State$, $\steo\in\State+1$. Then:
\begin{enumerate}
\item If $\steo=\top(\sts)$ then $\appr{\pgg}{\pee}{0}$;
\item If $\steo=\top(\sts')$ for some $\sts'\ne\sts$ then
$\appr{\pgg}{\pee}{1}$;
\item If $\steo=\bot$ then $\appr{\pgg}{\pee}{-1}$.
\end{enumerate}
\end{theorem}

\begin{proof}
By induction on the shape of the derivation tree of the assertion
$\parse{\pgg}{\pee}{\rlt}{\sts}{\steo}$.
\end{proof}

\term{Well-formedness} is a decidable conservative 
approximation of the predicate that is true iff parsing in $\pgg$ never
diverges (it definitely excludes grammars with left recursion but can  
exclude also some safe grammars). Well-formedness of PEGs was introduced by
Ford~\cite{DBLP:conf/popl/Ford04}. 
The following set of clauses is an inductive definition of predicate
$\WF_{\pgg}$, well-formedness of expressions, for \npeg:
\begin{enumerate}
\item $\e\in\WF_{\pgg}$;
\item For every $\tma\in\tmst$, $\tma\in\WF_{\pgg}$;
\item For every $\ntx\in\ntsn$, $\ntx\in\WF_{\pgg}$ if
$\pefd(\ntx)\in\WF_{\pgg}$;
\item For every $\pep,\peq\in\Expr{\pgg}$, $\pep\peq\in\WF_{\pgg}$ if 
$\pep\in\WF_{\pgg}$ and, in addition, $\appr{\pgg}{\pep}{0}$ implies
$\peq\in\WF_{\pgg}$;
\item For every $\pep,\peq\in\Expr{\pgg}$, $\pep\alt\peq\in\WF_{\pgg}$ if
$\pep\in\WF_{\pgg}$ and, in addition, $\appr{\pgg}{\pep}{-1}$ implies
$\peq\in\WF_{\pgg}$;
\item For every $\pep\in\Expr{\pgg}$, $\lkh{\pep}\in\WF_{\pgg}$ if
$\pep\in\WF_{\pgg}$;
\item For every $\pep\in\Expr{\pgg}$, $\rep{\pep}\in\WF_{\pgg}$ if
$\nappr{\pgg}{\pep}{0}$ and $\pep\in\WF_{\pgg}$;
\item For every $\pep\in\Expr{\pgg}$ and $\rlr\in\Rel^+(\Nat)$,
$\ind{\rlr}{\pep}\in\WF_{\pgg}$ if $\pep\in\WF_{\pgg}$;
\item For every $\pep\in\Expr{\pgg}$ and $\rls\in\Rel^+(\Nat)$,
$\loc{\rls}{\pep}\in\WF_{\pgg}$ if $\pep\in\WF_{\pgg}$;
\item For every $\pep\in\Expr{\pgg}$, $\aln{\pep}\in\WF_{\pgg}$ if
$\pep\in\WF_{\pgg}$.
\end{enumerate}

This definition rejects non-terminals with directly or indirectly left 
recursive rules since for a concatenation $\pep\peq$ to be well-formed, 
$\pep$ must be well-formed, leading to an infinite derivation 
in the case of any kind of left recursion. On the other hand, 
requiring both 
$\pep\in\WF_{\pgg}$ and $\peq\in\WF_{\pgg}$ in the clause for 
$\pep\peq\in\WF_{\pgg}$ would be too restrictive since this would reject 
non-terminals
with meaningful recursive productions like $\ntx\mapsto\tma\ntx\alt\e$. 

Again, clauses for PEG constructs (1--7) mostly copy the definition given by 
Ford~\cite{DBLP:conf/popl/Ford04}. This time, the choice case is an
exception. In
\cite{DBLP:conf/popl/Ford04}, $\pep\alt\peq$ is considered well-formed only
if both $\pep$ and $\peq$ are well-formed, which needlessly rejects 
non-terminals with safe recursive rules like $\ntx\mapsto\e\alt\ntx$. We require
$\peq\in\WF_{\pgg}$ only if $\peq$ could possibly be executed, i.e. if
$\appr{\pgg}{\pep}{-1}$.

A grammar $\pgg=(\ntsn,\tmst,\pefd,\pes)$ is called \term{well-formed} if 
$\pep\in\WF_{\pgg}$ for every expression $\pep$ occurring as a subexpression 
in some $\pefd(\ntx)$ or $\pes$. 
Ford \cite{DBLP:conf/popl/Ford04} proves by induction on the length of the 
input string that, in well-formed grammars, parsing of every expression whose 
all subexpressions are well-formed terminates on every input string. We can
prove an analogous result in a similar way but we prefer to generalize 
the statement to a stricter
semantics which enables to occasionally construct easier proofs later. 
The new semantics, which we call \term{strict}, is defined by replacing
the choice clause in the definition of Subsect.~\ref{sem:adams} 
with the following:
\begin{enumerate}\setcounter{enumi}{4}
\item For every $\pep,\peq\in\Expr{\pgg}$,
$\parse{\pgg}{\pep\alt\peq}{\rlt}{\sts}{\steo}$ holds in two cases:
\begin{itemize}
\item There exists a triple $\sts'$ such that
$\parse{\pgg}{\pep}{\rlt}{\sts}{\top(\sts')}$, $o=\top(\sts')$ and, in
addition, $\appr{\pgg}{\pep}{-1}$ implies
$\parse{\pgg}{\peq}{\rlt}{\sts}{\steo'}$ for some
$\steo'\in\State+1$;
\item $\parse{\pgg}{\pep}{\rlt}{\sts}{\bot}$ and
$\parse{\pgg}{\peq}{\rlt}{\sts}{\steo}$.
\end{itemize}
\end{enumerate}

The new semantics is more restrictive since, to finish 
parsing of an expression of the form $\pep\alt\peq$ 
after parsing $\pep$ successfully, also $\peq$ must be parsed 
if $\appr{\pgg}{\pep}{-1}$ happens to be the case. In the standard
semantics, parsing of $\pep\alt\peq$ does not have to try $\peq$ if parsing
of $\pep$ is successful. So, if parsing of an expression terminates in the
strict semantics then it terminates with the same result in the standard 
semantics (but not necessarily vice versa). Therefore proving that parsing
always gives a result in the strict semantics will establish this also for
the standard semantics. In the rest, we sign strict semantics with
exclamation mark, i.e., parsing assertions will be of the form
$\sparse{\pgg}{\pee}{\rlt}{\sts}{\steo}$. 

\begin{theorem}\label{elim:wf:strict}
Let $\pgg=(\ntsn,\tmst,\pefd,\pes)$ be a well-formed \npeg\ and let
$\pee\in\Expr{\pgg}$. Assume that all subexpressions of $\pee$ are 
well-formed. Then for every $\rlt\in\Rel^+(\Nat)$ and
$\sts\in\State$, there exists $\steo\in\State+1$ such that 
$\sparse{\pgg}{\pee}{\rlt}{\sts}{\steo}$.
\end{theorem}

\begin{proof}
By induction on the length of the input string (i.e., the first 
component of $\sts$). 
The induction step uses induction on the shape of the derivation tree of 
the assertion $\pee\in\WF_{\pgg}$.
\end{proof}

\subsection{Splitting}\label{trans:split}

As the repetition operator can always be eliminated (by adding a new 
non-terminal $\nta_{\pep}$ with $\pefd(\nta_{\pep})=\pep\nta_{\pep}\alt\e$ 
for each subexpression $\pep$ that occurs under the star operator), 
we may assume that the input grammar $\pgg$ is repetition-free. 
The first stage of our process also assumes that $\pgg$ is
well-formed, all negations are applied to atomic expressions, and all
choices are disjoint. A choice expression $\pep\alt\peq$ is called 
\term{disjoint} if parsing of $\pep$ and $\peq$ cannot succeed in the same 
input state and token mode. Achieving the last two 
preconditions can be considered as a preparatory and previously studied
(e.g. in \cite{DBLP:conf/popl/Ford04} as stage 1 of negation elimination)
step of the process. 
Issues concerning this are discussed briefly in 
Subsect.~\ref{trans:prelim}. 

We use in principle the same splitting algorithm as in stage 2 of the
negation elimination process described by 
Ford~\cite{DBLP:conf/popl/Ford04}, adding clauses for the extra operators
in \npeg. The approach defines two functions
$\pefg_0:\WF_{\pgg}\to\Expr{\pgg}$ and $\pefg_1:\Expr{\pgg}\to\Expr{\pgg}$
as follows ($F$ is a metavariable denoting any expression that always fails,
e.g., $\lkh\e$):
\[\setlength{\arraycolsep}{2pt}
\begin{array}{lcl@{\hspace{0.5em}}lcl}
\pefg_0(\e)&=&\e&\pefg_1(\e)&=&F\\
\pefg_0(\tma)&=&F&\pefg_1(\tma)&=&\tma\\
\pefg_0(\ntx)&=&\pefg_0(\pefd(\ntx))&\pefg_1(\ntx)&=&\ntx\\
\pefg_0(\pep\peq)&=&\left\lbrace
\begin{array}{ll}
\pefg_0(\pep)\pefg_0(\peq)&\mbox{if $\appr{\pgg}{\pep}{0}$}\\
F&\mbox{otherwise}\end{array}\right\rbrace&\pefg_1(\pep\peq)&=&\pefg_1(\pep)\pefg_1(\peq)\alt\pefg_1(\pep)\pefg_0(\peq)\alt\pefg_0(\pep)\pefg_1(\peq)\\
\pefg_0(\pep\alt\peq)&=&\left\lbrace
\begin{array}{ll}
\pefg_0(\pep)\alt\pefg_0(\peq)&\mbox{if $\appr{\pgg}{\pep}{-1}$}\\
\pefg_0(\pep)&\mbox{otherwise}
\end{array}\right\rbrace&\pefg_1(\pep\alt\peq)&=&\left\lbrace
\begin{array}{ll}
\pefg_1(\pep)\alt\pefg_1(\peq)&\mbox{if $\appr{\pgg}{\pep}{-1}$}\\
\pefg_1(\pep)&\mbox{otherwise}
\end{array}
\right\rbrace\\
\pefg_0(\lkh{\pep})&=&\lkh(\pefg_1(\pep)\alt\pefg_0(\pep))&\pefg_1(\lkh{\pep})&=&F\\
\pefg_0(\ind{\rlr}{\pep})&=&\ind{\rlr}{(\pefg_0(\pep))}&\pefg_1(\ind{\rlr}{\pep})&=&\ind{\rlr}{(\pefg_1(\pep))}\\
\pefg_0(\loc{\rls}{\pep})&=&\loc{\rls}{(\pefg_0(\pep))}&\pefg_1(\loc{\rls}{\pep})&=&\loc{\rls}{(\pefg_1(\pep))}\\
\pefg_0(\aln{\pep})&=&\aln{\pefg_0(\pep)}&\pefg_1(\aln{\pep})&=&\aln{\pefg_1(\pep)}
\end{array}
\]

Correctness of the definition of $\pefg_0$ follows by induction on the 
shape of the derivation tree of the assertion $\pee\in\WF_{\pgg}$. 
In the negation case, 
we use that negations are applied to atomic expressions, whence the reference 
to $\pefg_1$ can be eliminated by a replacement from its definition. The 
definition of $\pefg_1$ is sound by induction on the shape of the expression 
$\pee$. 

A new grammar $\pgg'=(\ntsn,\tmst,\pefd',\pes')$ is defined using
$\pefg_0,\pefg_1$ by equations $\pefd'(\ntx)=\pefg_1(\pefd(\ntx))$,
$\pes'=\pefg_1(\pes)\alt\pefg_0(\pes)$. 
The equivalence of the input and output grammars relies on the splitting 
invariant established by Theorem~\ref{trans:split:inv} below
which allows instead of each parsing expression~$\pee$ 
with negations in front of atoms and disjoint choices in $\pgg$ 
to equivalently use parsing expression $\pefg_1(\pee)\alt\pefg_0(\pee)$
in $\pgg'$. The claim is analogous to the splitting invariant 
used by \cite{DBLP:conf/popl/Ford04} but we can provide a simpler proof
using the strict semantics (an analogous proof using the standard semantics 
would fail in the choice case).

\begin{theorem}\label{trans:split:inv}
Let $\parse{\pgg}{\pee}{\rlt}{\sts}{\steo}$ where $\pee\in\Expr{\pgg}$, 
$\rlt\in\Rel^+(\Nat)$ and $\sts\in\State$, $\steo\in\State+1$. 
Assuming that all choices in the rules of $\pgg$ and expression $\pee$
are disjoint and the negations are applied to atoms, the following holds:
\begin{enumerate}
\item If $\steo=\top(\sts)$ then
$\parse{\pgg'}{\pefg_0(\pee)}{\rlt}{\sts}{\top(\sts)}$ and 
$\parse{\pgg'}{\pefg_1(\pee)}{\rlt}{\sts}{\bot}$;
\item If $\steo=\top(\sts')$ where $\sts'\ne\sts$ then
$\parse{\pgg'}{\pefg_0(\pee)}{\rlt}{\sts}{\bot}$ and 
$\parse{\pgg'}{\pefg_1(\pee)}{\rlt}{\sts}{\top(\sts')}$;
\item If $\steo=\bot$ then
$\parse{\pgg'}{\pefg_0(\pee)}{\rlt}{\sts}{\bot}$ and 
$\parse{\pgg'}{\pefg_1(\pee)}{\rlt}{\sts}{\bot}$.
\end{enumerate}
\end{theorem}

\begin{proof}
We don't use the repetition operator, whence all expressions in well-formed
grammars are well-formed (this fact follows from an easy induction on the 
expression structure). 
By Theorems~\ref{elim:wf:strict} and \ref{sem:uniq}, 
$\sparse{\pgg}{\pee}{\rlt}{\sts}{\steo}$. 
The desired result follows by induction of the shape of the derivation tree of 
$\sparse{\pgg}{\pee}{\rlt}{\sts}{\steo}$, using the disjointness assumption 
in the choice case.
\end{proof}

As the result of this transformation, 
the sizes of the right-hand sides of the productions can grow 
exponentially though the number of productions stays unchanged. 
Preprocessing the 
grammar via introducing new non-terminals in such a way that all 
concatenations were applied to atoms 
(similarly to Ford~\cite{DBLP:conf/popl/Ford04}) would hinder the growth, 
but the size in the worst case remains
exponential. The subsequent transformations cause 
at most a linear growth of right-hand sides.

\subsection{Alignment elimination}\label{trans:align}

In a grammar $\pgg=(\ntsn,\tmst,\pefd,\pes)$ 
obtained via splitting, we can eliminate alignments using the 
following three steps:
\begin{enumerate}
\item Introduce a copy $\ntx'$ of each non-terminal $\ntx$ and define 
$\pefd(\ntx')=\aln{\pefd(\ntx)}$.
\item In all right-hand sides of productions and the start expression, 
apply distributivity laws (Theorem~\ref{sem:maineq} (3), 
Theorem~\ref{sem:absorb} (3), Theorem~\ref{sem:alnhelp}) and idempotence 
(Theorem~\ref{sem:absorb} (4)) 
to bring all alignment operators down to terminals and non-terminals.
Replace alignment of terminals by position (Theorem~\ref{sem:absorb} (6)).
\item In all right-hand sides of productions and the start expression, 
replace all subexpressions of the form $\aln{\ntx}$ with the 
corresponding new non-terminal $\ntx'$.
\end{enumerate}

For establishing the equivalence of the original and the obtained grammar, 
the following general theorem can be used.

\begin{theorem}\label{trans:align:gen}
Let $\pgg_1=(\ntsn,\tmst,\pefd_1,\pes_1)$ and 
$\pgg_2=(\ntsn,\tmst,\pefd_2,\pes_2)$ be \npeg s. If for every
$\ntx\in\ntsn$, $\pefd_1(\ntx)\sim_{\pgg_1}\pefd_2(\ntx)$ then 
$\parse{\pgg_2}{\pee}{\rlt}{\sts}{\steo}$ always implies 
$\parse{\pgg_1}{\pee}{\rlt}{\sts}{\steo}$.
\end{theorem}

\begin{proof}
Easy induction on the shape of the derivation tree of
$\parse{\pgg_2}{\pee}{\rlt}{\sts}{\steo}$.
\end{proof}

Denote by $\peff$ the function defined on $\Expr{\pgg}$ that performs
transformations of steps 2--3, i.e., distributes alignment operators to the
non-terminals and replaces aligned non-terminals with corresponding new
non-terminals. Denote by $\pgg'$ the grammar obtained after step~3. 
Note that step~1 does not change the semantics of expressions written in
the original grammar. Steps 2 and 3 replace the right-hand sides of
productions with expressions that are semantically equivalent with them in
the grammar obtained after step~1. By Theorem~\ref{trans:align:gen}, 
this implies that whenever parsing of some $\pee\in\Expr{\pgg}$ in
the final grammar $\pgg'$ produces some result then the same result is 
obtained when parsing $\pee$ with the same input state and token mode 
in the original grammar~$\pgg$. In order to be able to apply 
Theorem~\ref{trans:align:gen} with grammars interchanged, 
we need the equivalence of the right-hand sides of productions also in 
grammar $\pgg'$. For this, it is sufficient to show 
$\aln{\ntx}\sim_{\pgg'}\ntx'$ for every $\ntx\in\ntsn$, which in turn
would follow from the statement 
$\aln{\peff(\pefd(\ntx))}\sim_{\pgg'}\peff(\aln{\pefd(\ntx)})$. Consequently, 
the equivalence of the initial and final grammars is implied by 
the following theorem.

\begin{theorem}\label{trans:align:thm}
For every $\pee\in\Expr{\pgg}$,
$\peff(\aln{\pee})\sim_{\pgg'}\aln{\peff(\pee)}$. 
\end{theorem}

\begin{proof}
The claim is a direct consequence of the following two lemmas, both holding 
for arbitrary $\sts\in\State$, $\steo\in\State+1$:
\begin{itemize}
\item If $\parse{\pgg'}{\peff(\aln{\pee})}{\rlt}{\sts}{\steo}$ then
$\parse{\pgg'}{\aln{\peff(\pee)}}{\rlt}{\sts}{\steo}$ and
$\parse{\pgg'}{\baln{\peff(\aln{\pee})}}{\rlt}{\sts}{\steo}$;
\item If $\parse{\pgg'}{\aln{\peff(\pee)}}{\rlt}{\sts}{\steo}$ or
$\parse{\pgg'}{\baln{\peff(\aln{\pee})}}{\rlt}{\sts}{\steo}$ then 
$\parse{\pgg'}{\peff(\aln{\pee})}{\rlt}{\sts}{\steo}$.
\end{itemize}
Both lemmas are proven by induction on the shape of derivation trees.
The assertion with two alignments (both outside and inside) is needed in the
case where $\pee$ itself is of the form~$\aln{\pep}$.
\end{proof}

\subsection{Elimination of position operators}\label{trans:loc}

In an alignment-free \npeg\ $\pgg=(\ntsn,\tmst,\pefd,\pes)$, 
we can get rid of position operations using a process largely 
analogous to the alignment elimination, consisting of the following 
four steps:
\begin{enumerate}
\item Introduce a new non-terminal $\nnter{\ntx,\rlt}$ 
for each existing non-terminal $\ntx$ and relation $\rlt$ used by a position
operator, with 
$\pefd(\nnter{\ntx,\rlt})=\loc{\rlt}{(\pefd(\ntx))}$.
\item Apply distributivity laws (Theorem \ref{sem:maineq} (1)) and
cancellation (Theorem \ref{sem:absorb} (5)) to bring all position operators
down to terminals and non-terminals.
\item Replace all subexpressions of the form $\loc{\rlt}{\ntx}$ with
corresponding new non-terminals $\nnter{\ntx,\rlt}$.
\item Replace all subexpressions of the form $\loc{\rlt}{\tma}$ with
$\ind{\rlt}{\tma}$.
\end{enumerate}

Again, denote by $\peff$ the function defined on $\Expr{\pgg}$ that performs 
transformations of steps 2--3, i.e., distributes position operators to the 
terminals and non-terminals and replaces non-terminals under position
operators with 
corresponding new non-terminals. Denote by $\pgg'$ the grammar obtained after 
step 3. 
Theorem~\ref{trans:align:gen} applies here as well, whence the
equivalence of the grammar obtained after step 3 and the initial 
grammar is implied by the following Theorem~\ref{elim:loc:thm}.

\begin{theorem}\label{elim:loc:thm}
For 
every $\pee\in\Expr{\pgg}$ and $\rls\in\Rel^+(\Nat)$,
$\peff(\loc{\rls}{\pee})\sim_{\pgg'}\loc{\rls}{(\peff(\pee))}$.
\end{theorem}

\begin{proof}
The claim is a direct consequence of the following two lemmas, both holding 
for arbitrary $\sts\in\State$, $\steo\in\State+1$ and 
$\rlt,\rlu\in\Rel^+(\Nat)$:
\begin{itemize}
\item If
$\parse{\pgg'}{\peff(\loc{\rls}{\pee})}{\rlt}{\sts}{\steo}$ then
$\parse{\pgg'}{\loc{\rls}{(\peff(\pee))}}{\rlt}{\sts}{\steo}$ and 
$\parse{\pgg'}{\loc{\rlu}{(\peff(\loc{\rls}{\pee}))}}{\rlt}{\sts}{\steo}$;
\item If 
$\parse{\pgg'}{\loc{\rls}{(\peff(\pee))}}{\rlt}{\sts}{\steo}$ or 
$\parse{\pgg'}{\loc{\rlu}{(\peff(\loc{\rls}{\pee}))}}{\rlt}{\sts}{\steo}$ then
$\parse{\pgg'}{\peff(\loc{\rls}{\pee})}{\rlt}{\sts}{\steo}$.
\end{itemize}
Both lemmas are proven by induction on the shape of the derivation trees. 
The claim with position operator both outside and inside 
($\loc{\rlu}{(\peff(\loc{\rls}{\pee}))}$) is needed in the case when $\pee$ 
itself is an application of the position operator.
\end{proof}

Correctness of step 4 can be proven by induction on the shape of the 
derivation trees, using Theorem \ref{sem:locind}. Note that 
here we must assume that parsing according to the final grammar is
performed with the alignment flag false (a natural assumption as the grammar
is alignment-free) and the token mode $\idrel$.

\subsection{Discussion on the preconditions}\label{trans:prelim}

Alignment elimination was correctly defined under the assumption that the
input grammar is well-formed, has negations only in front of atoms, and 
disjoint choices (all these conditions are needed at stage 1 only). 
The second assumption can be easily established by introducing a new
non-terminal for each expression $\pep$ such that $\lkh{\pep}$ occurs in
the productions or in the start expression. This can be done in the lines of
the first stage of the negation elimination process described by 
Ford \cite{DBLP:conf/popl/Ford04}. This transformation preserves
well-formedness of the grammar. Achieving disjoint choices is a more
subtle topic. A straightforward way would be replacing choices of the form
$\pep\alt\peq$ with disjoint choices $\pep\alt\lkh{\pep}\peq$ which seems to
work well as $\pep\alt\peq$ and $\pep\alt\lkh{\pep}\peq$ are equivalent in
the standard semantics. 

Alas, $\pep\alt\peq$ and $\pep\alt\lkh{\pep}\peq$ are not equivalent in the
approximation semantics, because if $\appr{\pgg}{\pep}{1}$, 
$\appr{\pgg}{\pep}{-1}$, $\appr{\pgg}{\peq}{0}$ but 
$\nappr{\pgg}{\peq}{-1}$, then
$\appr{\pgg}{\pep\alt\lkh{\pep}\peq}{-1}$ while
$\nappr{\pgg}{\pep\alt\peq}{-1}$. Due to this, replacing $\pep\alt\peq$
with $\pep\alt\lkh{\pep}\peq$ can break well-formedness. Take
$\ntx\in\ntsn$ such that $\pefd(\ntx)=\lkh{(\tma\alt\e)}\ntx$. Then
$\ntx\in\WF_{\pgg}$ due to $\lkh{(\tma\alt\e)}\in\WF_{\pgg}$ alone, no 
recursive call to $\ntx\in\WF_{\pgg}$ arises as
$\nappr{\pgg}{\lkh{(\tma\alt\e)}}{0}$. However, if
$\pefd'(\ntx)=\lkh{(\tma\alt\lkh{\tma}\e)}\ntx$ in $\pgg'$ then 
$\appr{\pgg'}{\lkh{(\tma\alt\lkh{\tma}{\e})}}{0}$ whence well-formedness of
$\ntx$ now recursively requires well-formedness of $\ntx$. Thus
$\ntx\nin\WF_{\pgg'}$. (An argument similar to this shows that the
first stage of the negation elimination process in 
Ford \cite{DBLP:conf/popl/Ford04} also can break well-formedness. As the
second stage is correctly defined only for well-formed grammars, the whole
process fails.)

One solution would be changing the approximation semantics by adding, to the
inductive definition in Subsect.~\ref{trans:wf}, a
general clause
\begin{enumerate}\setcounter{enumi}{-1}
\item $\appr{\pgg}{\pee}{-1}$ if $\appr{\pgg}{\pee}{0}$ or
$\appr{\pgg}{\pee}{1}$.
\end{enumerate}
This forces $\appr{\pgg}{\pee}{-1}$ to hold whenever an assertion of the
form $\appr{\pgg}{\pee}{n}$ holds, and in particular, $\pep\alt\peq$ becomes
equivalent to $\pep\alt\lkh{\pep}\peq$. Then replacing $\pep\alt\peq$ with 
$\pep\alt\lkh{\pep}\peq$ preserves well-formedness. Although well-formedness
predicate becomes more restrictive and rejects more safe grammars, 
the loss seems to be little and acceptable in practice (expressions
$\peq$ such that $\appr{\pgg}{\peq}{0}$ or $\appr{\pgg}{\peq}{1}$ while
$\nappr{\pgg}{\peq}{-1}$ seem to occur not very commonly in influenced 
productions such as $\ntx\mapsto\lkh{(\pep\alt\peq)}\ntx$, but a further 
investigation is needed to clarify this).

\section{Related work}\label{other}

PEGs were first introduced and studied by 
Ford~\cite{DBLP:conf/popl/Ford04} who also showed them to be closely related 
with the TS system \cite{DBLP:journals/iandc/BirmanU73} and TDPL 
\cite{Aho:1972:TPT:578789}, as well as to their generalized forms
\cite{DBLP:journals/iandc/BirmanU73,Aho:1972:TPT:578789}.

Adams \cite{DBLP:conf/popl/Adams13} and Adams and A\uu{g}acan
\cite{DBLP:conf/haskell/AdamsA14} provide an excellent overview of 
previous approaches to describing indentation-sensitive languages and 
attempts of building indentation features into parser libraries. 
Our work is a theoretical study of the approach proposed in 
\cite{DBLP:conf/haskell/AdamsA14} while some details of the semantics used 
in our paper were \qquot{corrected} in the lines of Adams' 
\texttt{indentation} package for Haskell \cite{implementation}. 
This package enables specifying indentation 
sensitivity 
within the Parsec and Trifecta parser combinator libraries. 
A process of alignment operator elimination is previously described for CFGs
by Adams \cite{DBLP:conf/popl/Adams13}.

Matsumura and Kuramitsu~\cite{DBLP:journals/jip/MatsumuraK16} develop a very
general extension of PEG that also enables to specify indentation. Their 
framework is powerful but complicated. The approach proposed in
\cite{DBLP:conf/haskell/AdamsA14} and followed by us is in contrast with 
\cite{DBLP:journals/jip/MatsumuraK16} by focusing on indentation and aiming
to maximal simplicity and convenience of usage.

\section{Conclusion}\label{conc}

We studied the extension of PEG proposed by Adams and
A\uu{g}acan~\cite{DBLP:conf/haskell/AdamsA14} for indentation-sensitive
parsing. This extension uses operators for marking indentation and alignment 
besides the classic ones. Having added one more operator (position) for 
convenience,
we found a lot of useful semantic equivalences that are valid on expressions
written in the extended grammars. We applied these equivalences subsequently
for defining a process that algorithmically eliminates all alignment and
position operators from well-formed grammars.

\bibliography{nestra}

\end{document}